\documentclass[conference]{IEEEtran}
\usepackage{amsmath,amssymb,amsfonts,amsthm,bm}
\usepackage{graphicx}
\usepackage{cite}
\usepackage{booktabs}
\usepackage{xcolor}
\usepackage{hyperref}
\hypersetup{hidelinks}
\usepackage{pgfplots}
\pgfplotsset{compat=1.18}
\pgfplotsset{every axis plot/.append style={mark options={solid}}}
\usetikzlibrary{arrows.meta,patterns,decorations.pathmorphing,calc}

\newtheorem{proposition}{Proposition}
\newtheorem{lemma}{Lemma}

\title{Optical ISAC with Single-Photon Avalanche Diode (SPAD) Arrays}
\author{\IEEEauthorblockN{Vasilis K. Papanikolaou\IEEEauthorrefmark{1}, Dimitris Bozanis\IEEEauthorrefmark{2}, Ferdaous Tarhouni\IEEEauthorrefmark{1}, \\ George K. Karagiannidis\IEEEauthorrefmark{2}, Christos Masouros\IEEEauthorrefmark{3}, and Robert Schober\IEEEauthorrefmark{1}}
\IEEEauthorblockA{\IEEEauthorrefmark{1}Institute for Digital Communications, Friedrich-Alexander-Universit\"at Erlangen-N\"urnberg (FAU), Erlangen, Germany
}
\IEEEauthorblockA{\IEEEauthorrefmark{2}Department of Electrical and Computer Engineering, Aristotle University of Thessaloniki, Thessaloniki, Greece}
\IEEEauthorblockA{\IEEEauthorrefmark{3}Department of Electronic and Electrical Engineering, University College London, London, U.K.}
\IEEEauthorblockA{e-mail: vasilis.papanikolaou@fau.de}
\vspace{-1cm}
}

\begin{document}
\maketitle

\begin{abstract}
Single-photon avalanche diode (SPAD) arrays offer photon-level sensitivity for sensing in optical integrated sensing and communication (O-ISAC). However, after each detection, a SPAD is blind for a dead time, which saturates the array at high received power. Existing O-ISAC analyses assume linear receivers and do not capture this effect. In this paper, we study ranging with a SPAD array in a monostatic O-ISAC system based on direct-current-biased optical orthogonal frequency-division multiplexing (DCO-OFDM). Sampling the array once per dead time, we model the number of pixels that detect a photon in each sampling interval as a binomial random variable whose success probability accounts for the intensity variation of the echo within a dead time. Based on this model, we derive the Cram\'er-Rao bound (CRB) for range estimation. We show that the SPAD count model of optical wireless communication, which assumes a constant photon rate over a dead time, causes a range bias that depends on the received power. 
Moreover, unlike for a linear photodiode, saturation degrades the ranging accuracy beyond a sensing-optimal transmit power. The rate-constrained operating point is then the larger of this power and the power needed for the target data rate. Numerical results validate the analysis against a continuous-time photon-counting simulation and show that the proposed model is needed for ranging.
\end{abstract}


\section{Introduction}\label{sec:intro}
Optical integrated sensing and communication (O-ISAC) recently emerged as the optical wireless counterpart of integrated sensing and communication (ISAC) in radio-frequency (RF) systems, reusing one intensity-modulated waveform for data transmission and environment sensing~\cite{wen2023fsoisac,nikolaidou2025oisac}. Direct-current-biased optical orthogonal frequency-division multiplexing (DCO-OFDM) supports both high-rate communication and delay-based sensing~\cite{nikolaidou2025oisac}. As in RF ISAC, the sensing performance of such a system is characterized by the Cram\'er-Rao bound (CRB) and the CRB-rate region~\cite{liu2018dfrc,liu2022crb,xiong2023fundamental}. 

Corner-cube reflectors (CCRs) strengthen the echo of cooperative targets~\cite{cui2024retroisac}. Diffuse targets, in contrast, re-radiate the intercepted light as a Lambertian source, so that two spreading losses compound into a $1/d^4$ round-trip loss at distance $d$~\cite{zhang2025led,komine2004vlc}. The low echo power motivates single-photon avalanche diode (SPAD) arrays as the sensing receiver~\cite{huang2023jocn}. Prior O-ISAC work, however, employs linear receivers~\cite{wen2023fsoisac,nikolaidou2025oisac,cui2024retroisac,zhang2025led,xu2025optisac}, and how the nonlinearity of a photon-counting front end affects the sensing-communication trade-off has not been analyzed.

In SPAD-based optical wireless communication (OWC), the dead time compresses and eventually paralyzes the response and makes the counts sub-Poisson~\cite{huang2023jocn,sarbazi2017jlt}. In single-photon light detection and ranging (LiDAR), pile-up~\cite{coates1968,rapp2021highflux} and the resulting range bias, similar to the range-walk error of pulsed LiDAR, are compensated by dead-time-aware estimators~\cite{rapp2019deadtime,wu2025spadcrlb}. For pulsed SPAD LiDAR, an optimal incident photon rate is also known, both from an attenuation rule~\cite{gupta2019flooded} and from a CRB under dead time~\cite{wu2025spadcrlb}. None of this transfers directly to ISAC, where the photon rate follows a data-bearing DCO-OFDM intensity rather than a pulse train. As a result, the sensing optimum is coupled with the rate through the shared transmit power. Moreover, the SPAD count models of OWC~\cite{huang2023jocn} assume a photon rate that is quasi-static over a dead time. Ranging, however, demands a sampling interval as short as the dead time, where this assumption fails and the pile-up correction must be used instead~\cite{coates1968,rapp2021highflux}.

Motivated by the above, we analyze ranging and the rate-constrained transmit power of a monostatic DCO-OFDM O-ISAC system that combines the sensing channels of CCRs and diffuse targets~\cite{cui2024retroisac,zhang2025led} with the passive-quenching SPAD-array model of~\cite{huang2023jocn}. We formulate the pile-up-corrected ranging observation and derive the joint delay-amplitude Fisher information and the range CRB. Through simulations, we show that the maximum-likelihood (ML) estimator based on this model is unbiased and, up to moderate saturation, close to the bound, whereas an estimator based on the quasi-static model exhibits a power-dependent range bias. Finally, we show that, once the echo saturates it, the ranging accuracy of the SPAD array degrades beyond a sensing-optimal transmit power, unlike for a linear photodiode (PD) receiver, which benefits from more power in both functions.

\section{System Model}\label{sec:system}

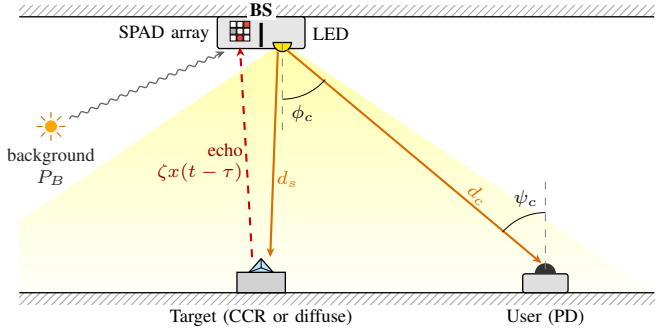
\begin{figure}[t]
\centering
\begin{tikzpicture}[
  font=\scriptsize, >={Stealth[length=3.5pt,width=3pt]}, line width=0.5pt,
  lab/.style={inner sep=1pt},
]
\def\yc{4.2}\def\yf{0.55}
\coordinate (spad) at (2.92,3.86);
\coordinate (led)  at (3.48,3.86);
\coordinate (pd)   at (6.96,0.86);
\begin{scope}
\clip (0,\yf) rectangle (8.4,\yc);
\shade[top color=yellow!55, bottom color=yellow!4, opacity=0.6] (led) -- (-1.4,\yf) -- (8.4,\yf) -- cycle;
\end{scope}
\fill[pattern=north east lines, pattern color=gray!70] (0,\yc) rectangle (8.4,\yc+0.15);
\draw (0,\yc) -- (8.4,\yc);
\fill[pattern=north east lines, pattern color=gray!70] (0,\yf-0.15) rectangle (8.4,\yf);
\draw (0,\yf) -- (8.4,\yf);
\draw[fill=gray!20, rounded corners=1pt] (2.62,3.78) rectangle (3.78,\yc);
\foreach \i in {0,1,2} \foreach \j in {0,1,2} {
  \pgfmathsetmacro{\c}{(\i==1 && \j==0) || (\i==2 && \j==2) ? "red!75" : ((\i==0 && \j==1) ? "gray!60" : "white")}
  \draw[fill=\c, line width=0.2pt] ($(spad)+(-0.13+0.09*\i,0.02+0.09*\j)$) rectangle ++(0.08,0.08);
}
\draw[line width=1.2pt] (3.2,3.8) -- (3.2,4.12);
\fill[yellow!90!orange] ($(led)+(-0.12,0)$) arc[start angle=180, end angle=360, radius=0.12] -- cycle;
\draw ($(led)+(-0.12,0)$) arc[start angle=180, end angle=360, radius=0.12] -- cycle;
\node[fill=white, inner sep=1pt, font=\scriptsize\bfseries] at (3.2,4.29) {BS};
\node[lab, anchor=east] at (2.56,3.98) {SPAD array};
\node[lab, anchor=west] at (3.84,3.98) {LED};
\draw[fill=gray!35] (2.88,\yf) rectangle (3.52,0.82);
\draw[fill=cyan!25] (3.04,0.82) -- (3.36,0.82) -- (3.2,1.0) -- cycle;
\draw[line width=0.25pt] (3.2,0.88) -- (3.04,0.82) (3.2,0.88) -- (3.36,0.82) (3.2,0.88) -- (3.2,1.0);
\node[lab, anchor=north] at (3.2,\yf-0.17) {Target (CCR or diffuse)};
\draw[fill=gray!25, rounded corners=1pt] (6.66,\yf) rectangle (7.26,0.8);
\fill[black!80] (6.82,0.8) arc[start angle=180, end angle=0, radius=0.14] -- cycle;
\node[lab, anchor=north] at (6.96,\yf-0.17) {User (PD)};
\draw[gray, dashed, line width=0.3pt] (led) -- ($(led)+(0,-1.25)$);
\draw[gray, dashed, line width=0.3pt] (pd) -- ($(pd)+(0,1.25)$);
\draw[->, thick, orange!85!black] ($(led)+(0.06,-0.1)$) -- node[lab, sloped, above, pos=0.72]{$d_c$} ($(pd)+(-0.05,0.08)$);
\draw[line width=0.3pt] ($(led)+(0,-0.75)$) arc[start angle=-90, end angle=-40.8, radius=0.75];
\node[lab] at ($(led)+(0.27,-0.94)$) {$\phi_c$};
\draw[line width=0.3pt] ($(pd)+(0,0.75)$) arc[start angle=90, end angle=139.2, radius=0.75];
\node[lab] at ($(pd)+(-0.26,0.94)$) {$\psi_c$};
\draw[->, thick, orange!85!black] ($(led)+(-0.06,-0.12)$) -- node[right, lab, pos=0.62]{$d_s$} (3.32,1.02);
\draw[->, thick, red!80!black, dashed] (3.08,1.02) -- node[left, lab, pos=0.45, align=right, text=red!60!black]{echo\\$\zeta x(t-\tau)$} ($(spad)+(0,-0.06)$);
\fill[orange!70!yellow] (0.42,2.75) circle (0.08);
\foreach \a in {0,45,...,315} \draw[orange!80!black, line width=0.4pt] ($(0.42,2.75)+(\a:0.11)$) -- ($(0.42,2.75)+(\a:0.19)$);
\node[lab, align=center, text=gray!30!black, anchor=north] at (0.42,2.5) {background\\$P_B$};
\draw[->, gray!80!black, decorate, decoration={snake, amplitude=0.6pt, segment length=4pt, post length=3pt}] (0.66,2.85) -- ($(spad)+(-0.2,-0.12)$);
\end{tikzpicture}
\caption{Considered O-ISAC system.}
\label{fig:system}
\end{figure}

We consider the system in Fig.~\ref{fig:system}, with a dual-function base station (BS) mounted on the ceiling and facing downwards, which carries a light-emitting diode (LED) transmitter and a co-located SPAD array with $N_a$ pixels. A communication user at line-of-sight range $d_c$ is equipped with a conventional PD facing upwards, and a passive target at range $d_s$ on the axis of the BS reflects the transmitted waveform back to it~\cite{nikolaidou2025oisac}. The direct LED-to-array leakage 
is assumed to be optically suppressed, with its residual absorbed in the background rate, as in the monostatic O-ISAC literature~\cite{nikolaidou2025oisac,cui2024retroisac}.

\subsection{DCO-OFDM Waveform and Communication Link}

A DCO-OFDM block consists of the $N$ intensity samples~\cite{huang2023jocn} 
\begin{equation}
x[n]=\bigl(s[n]+B\bigr)^+,
\label{eq:dcoofdm}
\end{equation}
for $n=0,\dots,N-1$, where $s[n]$ is the real-valued multicarrier signal obtained from $N$ Hermitian-symmetric subcarriers with $M$-ary quadrature amplitude modulation ($M$-QAM). The bias is $B=\beta\sigma_s$, with $\sigma_s$ being the standard deviation of $s[n]$ and $\beta$ the normalized bias\footnote{Since $s[n]$ is approximately Gaussian~\cite{huang2023jocn}, it is clipped with probability $Q(\beta)$, with $Q(\cdot)$ denoting the $Q$-function, and the clipping noise is neglected for sufficiently large $\beta$.}. The operator $(\cdot)^+=\max\{\cdot,0\}$ enforces the non-negativity of the intensity. The block is preceded by a cyclic prefix (CP) of $L_{\mathrm{cp}}$ samples, and consecutive blocks are transmitted back to back. The samples are spaced by $T_s$, and the LED emits the intensity $x(t)=(s(t)+B)^+$, with $s(t)$ being the band-limited interpolation of $\{s[n]\}$, so that $x(nT_s)=x[n]$. The average emitted optical power is $P_t=\mathbb E[x(t)]\approx B$, where $\mathbb E[\cdot]$ denotes the expectation over the data symbols. The PD user observes $y_c[n]=h_c x[n]+w_c[n]$. Here, $h_c=(m+1)A_{\mathrm{pd}}\rho_{\mathrm{pd}}\kappa_{\mathrm{opt}}\cos^m\!\phi_c\cos\psi_c/(2\pi d_c^2)$ is the Lambertian line-of-sight gain~\cite{komine2004vlc}, where $m=-\ln2/\ln\cos\Psi_{1/2}$ is the Lambertian order for half-power semi-angle $\Psi_{1/2}$, $\phi_c$ and $\psi_c$ are the irradiance and incidence angles, and $A_{\mathrm{pd}}$, $\rho_{\mathrm{pd}}$, and $\kappa_{\mathrm{opt}}$ denote the PD area, its responsivity, and the optical filter-concentrator gain, respectively. Moreover, $w_c[n]$ is the zero-mean additive white Gaussian noise of the thermal-noise-limited PD front end, with power spectral density $N_0$~\cite{cui2024retroisac} and variance $\sigma_c^2=N_0B_{\mathrm{occ}}$, where $B_{\mathrm{occ}}=(N/2-1)/(NT_s)$ is the occupied bandwidth. The electrical signal-to-noise ratio (SNR) per data subcarrier is thus $\mathrm{SNR}_c=h_c^2\sigma_s^2/\sigma_c^2$.

\subsection{Sensing Channel and Photon Rate}

The CCR and diffuse channel gains at range $d_s$ are given by~\cite{cui2024retroisac,zhang2025led}
\begin{equation}
\begin{aligned}
\zeta_{\mathrm{CCR}}&=\frac{(m+1)A_{\mathrm{spad}}k_{\mathrm{eff}}}{8\pi d_s^2}\cos^{m+1}\!\phi_s\cos\psi_s,\\
\zeta_{\mathrm{diff}}&=\frac{(m+1)A_{\mathrm{tgt}}\rho_{\mathrm{diff}}A_{\mathrm{spad}}}{2\pi^2 d_s^4}\cos^m\!\phi_s\cos\psi_t\cos\chi\cos\psi_s,
\end{aligned}
\label{eq:zeta}
\end{equation}
respectively, where $A_{\mathrm{spad}}$ is the array aperture, $k_{\mathrm{eff}}$ denotes the effective CCR reflectance, and $A_{\mathrm{tgt}}$ and $\rho_{\mathrm{diff}}$ are the area and the reflectivity of the diffuse target, respectively. In addition, $\phi_s$ and $\psi_s$ are the irradiance and incidence angles of the sensing link, and $\psi_t$ and $\chi$ denote the incidence and re-radiation angles at the diffuse target. Let $\tau=2d_s/c$ denote the round-trip delay, with speed of light $c$, and let $D=c\tau/2$ be the range to be estimated. The echo reaches the array with optical power $\zeta x(t-\tau)$, $\zeta\in\{\zeta_{\mathrm{CCR}},\zeta_{\mathrm{diff}}\}$, so that the incident photon rate is~\cite{huang2023jocn}
\begin{equation}
\lambda(t)=\lambda_b+\alpha_s\,x(t-\tau), 
\label{eq:lambda}
\end{equation}
where $\alpha_s=\zeta C_s$. Here, $C_s=\Upsilon_{\mathrm{PDE}}(1+\phi_{\mathrm{AP}}+\phi_{\mathrm{CT}})/E_{\mathrm{ph}}$ converts optical power into the photon rate, with $\Upsilon_{\mathrm{PDE}}$ being the photon detection efficiency, $\phi_{\mathrm{AP}}$ and $\phi_{\mathrm{CT}}$ the afterpulsing and crosstalk probabilities, and $E_{\mathrm{ph}}$ the photon energy at the operating wavelength. The rate $\lambda_b=(\vartheta_{\mathrm{DCR}}+\Upsilon_{\mathrm{PDE}}P_B/E_{\mathrm{ph}})(1+\phi_{\mathrm{AP}}+\phi_{\mathrm{CT}})$ collects the dark counts of the array, for total rate $\vartheta_{\mathrm{DCR}}$ and background light power $P_B$. As in~\cite{huang2023jocn}, afterpulsing and crosstalk are counted as additional independent arrivals, and the array is uniformly illuminated, so that the photon arrivals at each of its $N_a$ pixels form an independent Poisson process of rate $\lambda(t)/N_a$.

\subsection{SPAD Array and Pile-Up-Corrected Binomial Model}

After each detected photon, a SPAD pixel is quenched and remains blind to incident photons for a fixed interval, the dead time $\tau_d$~\cite{huang2023jocn}. With passive quenching, the standard low-complexity array implementation, every photon arriving during the dead time restarts it, i.e., the dead time is paralyzable, the more severe form of saturation~\cite{huang2023jocn,sarbazi2017jlt}. For a constant photon rate $\lambda$, the mean detected photon count of the array per sampling interval $T_s\ge\tau_d$ and its variance are~\cite{huang2023jocn} 
\begin{align}
\mu_a(\lambda)&=\lambda T_s e^{-\lambda\tau_d/N_a},
\label{eq:spad_mean}\\
\sigma_a^2(\lambda)&=\lambda T_s e^{-\frac{\lambda\tau_d}{N_a}}-\frac{\lambda^2T_s\tau_d}{N_a}e^{-\frac{2\lambda\tau_d}{N_a}}\Bigl(2-\frac{\tau_d}{T_s}\Bigr),
\label{eq:spad_var}
\end{align}
respectively, which we refer to as the quasi-static model. The mean peaks at $\lambda^\star=N_a/\tau_d$ and decreases beyond it, i.e., the array paralyzes, and the variance is sub-Poisson~\cite{sarbazi2017jlt}. The normalized photon rate $\epsilon=\lambda\tau_d/N_a=\lambda/\lambda^\star$ measures the saturation, with $\epsilon=1$ at the peak.

The array counts continuously, and the count of a pixel in a sampling interval depends on the photons of the preceding dead time. At the start of a block, these belong to the echo that precedes it. We therefore choose the CP longer than the round-trip delay plus one dead time, $\tau+\tau_d\le L_{\mathrm{cp}}T_s$, and use the $N$ sampling intervals that follow it, $t\in[0,NT_s)$, as the ranging observation, so that their counts depend only on the data of the current block.

Ranging calls for a short sampling interval, as the range bin is $cT_s/2$, and we set it equal to the dead time, $T_s=\tau_d$. This is the shortest interval for which a dead time reaches back only into the preceding interval, and the longest for which a pixel detects at most one photon per interval. Let $y[n]$ denote the number of pixels that detect a photon in the $n$-th interval $[nT_s,(n+1)T_s)$, $n=0,\dots,N-1$, and $\tilde q[n]$ the probability that a given pixel does so, i.e., the firing probability. Since the pixels have independent and identical arrival processes, the array count follows the binomial distribution with $N_a$ trials and success probability $\tilde q[n]$, which we write as
\begin{equation}
y[n]\sim\mathrm{Binomial}\bigl(N_a,\tilde q[n]\bigr).
\label{eq:binomial}
\end{equation}
A paralyzable pixel detects a photon arriving at $t$ only if no photon arrived in $(t-\tau_d,t)$~\cite{sarbazi2017jlt}. For Poisson arrivals of rate $\lambda(t)/N_a$, the probability of no arrival in $(t-\tau_d,t)$ is $\exp(-\frac{1}{N_a}\int_{t-\tau_d}^{t}\lambda(u)\,\mathrm{d}u)$. Since a pixel detects at most one photon per interval, its firing probability equals its mean number of detections in the interval, i.e., the pile-up integral~\cite{coates1968,rapp2021highflux}
\begin{equation}
\tilde q[n]=\frac{1}{N_a}\int_{nT_s}^{(n+1)T_s}\!\!\lambda(t)\exp\Bigl(-\tfrac{1}{N_a}\!\int_{t-\tau_d}^{t}\!\lambda(u)\,\mathrm{d}u\Bigr)\mathrm{d}t,
\label{eq:qmem}
\end{equation}
with $\lambda(t)$ given in \eqref{eq:lambda}. For a photon rate that is constant over the interval and the dead time before it, \eqref{eq:qmem} reduces to the quasi-static value $\mu_a(\lambda[n])/N_a$ of \eqref{eq:spad_mean}, where $\lambda[n]=\frac{1}{T_s}\int_{nT_s}^{(n+1)T_s}\lambda(t)\,\mathrm{d}t$ is the photon rate averaged over the $n$-th interval. With DCO-OFDM, at $T_s=\tau_d$, however, the rate varies within a dead time, and the exponent of \eqref{eq:qmem} weights the rate of the preceding interval. The two models thus differ from interval to interval, and an estimator based on \eqref{eq:spad_mean} incurs a power-dependent range bias~\cite{rapp2019deadtime,wu2025spadcrlb}.

\section{Sensing Performance and ISAC Operating Point}\label{sec:analysis}

\subsection{Range CRB}\label{sec:crb}

The delay and the reflected amplitude, $\bm\theta=[\tau,\alpha_s]^\top$, where $(\cdot)^\top$ denotes the transpose, are unknown and estimated jointly, whereas $\lambda_b$ is known, e.g., measured with the LED off. Successive counts are treated as independent, following the per-symbol observations of RF ISAC~\cite{liu2022crb,xiong2023fundamental}. This neglects the negative correlation between adjacent counts that the dead time induces.
\begin{lemma}
\label{lem:dq}
Let $\Lambda(t)=\int_{t-\tau_d}^{t}\lambda(u)\,\mathrm{d}u$ denote the mean number of photons incident on the array in $(t-\tau_d,t)$ and $r(t)=\lambda(t)e^{-\Lambda(t)/N_a}$ the detected-photon rate of the array, so that \eqref{eq:qmem} can be written as $\tilde q[n]=\frac{1}{N_a}\int_{nT_s}^{(n+1)T_s}r(t)\,\mathrm{d}t$. Then,
\begin{subequations}\label{eq:dq}
\begin{align}
\frac{\partial\tilde q[n]}{\partial\alpha_s}&=\frac{1}{N_a}\int_{nT_s}^{(n+1)T_s}\Bigl[\frac{x(t-\tau)}{\lambda(t)}\notag\\
&\qquad\qquad-\frac{1}{N_a}\int_{t-\tau_d}^{t}\!x(u-\tau)\,\mathrm{d}u\Bigr]r(t)\,\mathrm{d}t,\label{eq:dqa}\\
\frac{\partial\tilde q[n]}{\partial\tau}&=-\frac{r((n+1)T_s)-r(nT_s)}{N_a}.\label{eq:dqtau}
\end{align}
\end{subequations}
\end{lemma}
\begin{IEEEproof}
By \eqref{eq:lambda}, $\Lambda(t)=\lambda_b\tau_d+\alpha_s\int_{t-\tau_d}^{t}x(u-\tau)\,\mathrm{d}u$. Differentiating $r(t)$ with respect to $\alpha_s$ gives
\begin{equation*}
\frac{\partial r(t)}{\partial\alpha_s}=\Bigl[x(t-\tau)-\frac{\lambda(t)}{N_a}\int_{t-\tau_d}^{t}\!x(u-\tau)\,\mathrm{d}u\Bigr]e^{-\Lambda(t)/N_a},
\end{equation*}
which, with $e^{-\Lambda(t)/N_a}=r(t)/\lambda(t)$, is the integrand of \eqref{eq:dqa}, and integrating over the interval yields \eqref{eq:dqa}.

For \eqref{eq:dqtau}, $\lambda(t)$ and $\Lambda(t)=\lambda_b\tau_d+\alpha_s\int_{t-\tau-\tau_d}^{t-\tau}x(v)\,\mathrm{d}v$, and hence $r(t)$, depend on $t$ and $\tau$ only through $t-\tau$, i.e., $r(t)=r_0(t-\tau)$, where $r_0(\cdot)$ is the detected-photon rate for $\tau=0$. Thus, $\tilde q[n]=\frac{1}{N_a}\int_{nT_s-\tau}^{(n+1)T_s-\tau}r_0(v)\,\mathrm{d}v$, where $\tau$ appears only in the integration limits, and differentiating with respect to $\tau$ yields \eqref{eq:dqtau}.
\end{IEEEproof}

\begin{proposition}
\label{prop:fim}
Consider a coherent frame of $L$ consecutive blocks, over which $\bm\theta$ is constant, with independent counts \eqref{eq:binomial}. The Fisher information matrix (FIM) of $\bm\theta$, averaged over the data, is $\bar{\mathbf J}=L\,\mathbb E[\mathbf J]$, where $\mathbf J$ is the FIM of one block,
\begin{equation}
[\mathbf J]_{ij}=\sum_{n=0}^{N-1}\frac{N_a}{\tilde q[n](1-\tilde q[n])}\frac{\partial\tilde q[n]}{\partial\theta_i}\frac{\partial\tilde q[n]}{\partial\theta_j},
\label{eq:fim}
\end{equation}
for $i,j\in\{\tau,\alpha\}$, where $[\cdot]_{ij}$ denotes the $(i,j)$th entry of a matrix, $\theta_\tau=\tau$, and $\theta_\alpha=\alpha_s$. Here, $\tilde q[n]$ is given by \eqref{eq:qmem} and its derivatives by Lemma~\ref{lem:dq}.
With $(\cdot)^{-1}$ denoting the matrix inverse, the range CRB is
\begin{equation}
\mathrm{CRB}(D)=\Bigl(\frac{c}{2}\Bigr)^2\bigl[\bar{\mathbf J}^{-1}\bigr]_{\tau\tau}=\Bigl(\frac{c}{2}\Bigr)^2\frac{[\bar{\mathbf J}]_{\alpha\alpha}}{[\bar{\mathbf J}]_{\tau\tau}[\bar{\mathbf J}]_{\alpha\alpha}-[\bar{\mathbf J}]_{\tau\alpha}^2}.
\label{eq:crb}
\end{equation}
\end{proposition}
\begin{IEEEproof}
A $\mathrm{Binomial}(N_a,\tilde q)$ count carries the Fisher information $N_a/(\tilde q(1-\tilde q))$ about $\tilde q$~\cite{kay1993}, which translates into \eqref{eq:fim} through the derivatives of $\tilde q[n]$. For independent counts, the contributions of the $N$ samples and of the $L$ blocks add, and their average over the data symbols is $L\,\mathbb E[\mathbf J]$. Combining the above and $D=c\tau/2$ leads to \eqref{eq:crb}.
\end{IEEEproof}

Since the dead time correlates adjacent counts, \eqref{eq:crb} is the CRB of the independent-count model rather than a strict bound for the actual counts. Moreover, as $\bar{\mathbf J}$ averages the FIM over the data before inversion, \eqref{eq:crb} lower-bounds the CRB averaged over the data~\cite{xiong2023fundamental}.

In the quasi-static special case, $\tilde q[n]=\mu_a(\lambda[n])/N_a$ and, at $T_s=\tau_d$, $N_a\tilde q[n](1-\tilde q[n])=\sigma_a^2(\lambda[n])$. Hence, \eqref{eq:fim} becomes $[\mathbf J]_{ij}=\sum_{n}\Xi(\lambda[n])\,\partial\lambda[n]/\partial\theta_i\,\partial\lambda[n]/\partial\theta_j$ with $\Xi=\mu_a'^2/\sigma_a^2$ and $\mu_a'=\mathrm d\mu_a/\mathrm d\lambda=T_se^{-\epsilon}(1-\epsilon)$. Since $\lambda[n]$ averages $\lambda(t)=\lambda_b+\alpha_sx(t-\tau)$ over the $n$-th interval, $\partial\lambda[n]/\partial\tau=-[\lambda((n+1)T_s)-\lambda(nT_s)]/T_s$, as for \eqref{eq:dqtau}. In both models, the delay information thus enters through the change of a rate across the edges of the interval. In the quasi-static model, this is the incident rate $\lambda(t)$, scaled by the slope $\mu_a'$ of the mean response, and in \eqref{eq:fim}, it is the detected-photon rate $r(t)$, cf.~\eqref{eq:dqtau}.

Since $\mu_a'$ vanishes at $\lambda^\star$ and $\Xi$ decays as $\epsilon e^{-\epsilon}$ beyond, a saturated sample carries vanishing delay information. For the PD benchmark of the same bare aperture, which observes $\rho_{\mathrm{pd}}\zeta x(nT_s-\tau)$ in thermal noise of variance $\sigma_c^2$, this Gaussian observation yields the classical linear delay bound~\cite{kay1993}, which is monotone in $P_t$.

\subsection{Rate-Constrained Operating Point}\label{sec:opt}

Communication and sensing share the average optical power $P_t\approx B$. Since $\sigma_s=P_t/\beta$, the electrical SNR is quadratic in the optical power, $\mathrm{SNR}_c=\gamma_cP_t^2$ with $\gamma_c=h_c^2/(\beta^2\sigma_c^2)$. On each data subcarrier, adaptive $M$-QAM at a target bit error rate $\mathrm{BER}_{\mathrm t}$ operates at the SNR gap $\Gamma=-\ln(5\,\mathrm{BER}_{\mathrm t})/1.5$ from capacity~\cite{goldsmith1997mqam}, so that the achievable rate is
\begin{equation}
R_c(P_t)=\frac{N}{N+L_{\mathrm{cp}}}B_{\mathrm{occ}}\log_2\Bigl(1+\frac{\gamma_cP_t^2}{\Gamma}\Bigr).
\label{eq:rate}
\end{equation}
The rate is strictly increasing in $P_t$, so that a target rate $R_0$ is equivalent to the power floor $P_t^{\mathrm c}(R_0)=\sqrt{\Gamma(2^{R_0(N+L_{\mathrm{cp}})/(NB_{\mathrm{occ}})}-1)/\gamma_c}$. Since the ranging precision is not monotone in $P_t$, the operating point is the power that minimizes $\mathrm{CRB}(D)$ over $0<P_t\le P_{\max}$, with LED power limit $P_{\max}$, subject to $R_c(P_t)\ge R_0$, which is feasible for $R_0\le R_c(P_{\max})$. Let $P_t^{\mathrm s}$ denote the minimizer without the rate constraint, the sensing-optimal power. If the CRB does not decrease on $[P_t^{\mathrm s},P_{\max}]$, the operating point over the feasible powers $[P_t^{\mathrm c}(R_0),P_{\max}]$ is
\begin{equation}
P_t^\star=\max\bigl\{P_t^{\mathrm s},\,P_t^{\mathrm c}(R_0)\bigr\},
\label{eq:pstar}
\end{equation}
i.e., the system is sensing-limited when $R_0\le R_c(P_t^{\mathrm s})$ and rate-limited otherwise. For a linear receiver or a non-saturating target, the sensing quality is monotone in the power, so that \eqref{eq:pstar} collapses to the maximum power.

\section{Numerical Results}\label{sec:results}

\begin{table}[t]
\centering
\caption{Simulation Parameters}\label{tab:params}
\scriptsize
\setlength{\tabcolsep}{3pt}
\renewcommand{\arraystretch}{0.85}
\begin{tabular}{llll}
\toprule
$\tau_d=T_s$ & 10 ns & $N_a$ & $\{256,1024,4096,8192\}$ \\
$\Upsilon_{\mathrm{PDE}}$ & 0.35 & $\vartheta_{\mathrm{DCR}}$, $P_B$ & 0.5 MHz, 10 nW \\
$\phi_{\mathrm{AP}}$, $\phi_{\mathrm{CT}}$ & 0.75\%, 2.5\% & $A_{\mathrm{spad}}=A_{\mathrm{pd}}$ & 1 cm$^2$ \\
$k_{\mathrm{eff}}$ & 0.92 & $\rho_{\mathrm{diff}}$, $A_{\mathrm{tgt}}$ & 0.7, 100 cm$^2$ \\
$N$, $L_{\mathrm{cp}}$, $L$ & 64, 8, 30 & $\beta$, $\Psi_{1/2}$ ($m$) & 3, $60^\circ$ (1) \\
$d_c$, $d_s$ & 3 m, 2 m & $N_0$, $\rho_{\mathrm{pd}}$ & $10^{-22}$ A$^2$/Hz, 0.25 A/W \\
Wavelength & 450 nm & $\kappa_{\mathrm{opt}}$, $\phi_c=\psi_c$ & 2.5, $48^\circ$ \\
$\mathrm{BER}_{\mathrm t}$ & $10^{-3}$ & $P_{\max}$ & 2 W \\
$\Gamma$ at $\mathrm{BER}_{\mathrm t}$ & 5.5 dB & $\lambda_b$ & $8.2\times10^{9}$ s$^{-1}$ \\
\bottomrule
\end{tabular}
\end{table}

The simulation parameters are listed in Table~\ref{tab:params} and follow~\cite{huang2023jocn,cui2024retroisac,zhang2025led,nikolaidou2025oisac}. The user's PD is $2$~m below the BS and $2.24$~m off its axis, i.e., $\phi_c=\psi_c=48^\circ$, and for $\beta=3$, the clipping probability $Q(\beta)\approx1.3\times10^{-3}$ is negligible. Figs.~\ref{fig:count}(b) and~\ref{fig:crb} are reported versus the received optical power $P_{\mathrm{rx}}=\zeta P_t$, or equivalently the mean normalized photon rate $\epsilon=(\lambda_b+C_sP_{\mathrm{rx}})\tau_d/N_a$. All markers are obtained by Monte Carlo (MC) simulation. For the SPAD, the photon arrivals at each pixel, including those during the CP of each block, are drawn from a Poisson process of rate $\lambda(t)/N_a$, a pixel detects a photon only if no photon arrived in the preceding $\tau_d$ seconds, and $y[n]$ is the number of pixels that detect a photon in the $n$-th interval. Each marker in Figs.~\ref{fig:count}(b)--\ref{fig:tradeoff} is the root-mean-square error (RMSE) of the range estimate, averaged over $200$ frames with new data. For each frame, the ML estimator maximizes the log-likelihood of the binomial counts \eqref{eq:binomial} of its $L$ blocks,
\begin{equation}
\hat{\bm\theta}=\arg\max_{\bm\theta}\sum_{l,n}\bigl[y_l[n]\ln\tilde q_l[n]+(N_a-y_l[n])\ln(1-\tilde q_l[n])\bigr],
\label{eq:ml}
\end{equation}
where $y_l[n]$ is the count in the $n$-th interval of the $l$-th block and $\tilde q_l[n]$ is its firing probability for the parameters $\bm\theta$, given by \eqref{eq:qmem} or, for the estimator based on the quasi-static model, by $\mu_a(\lambda_l[n])/N_a$ from \eqref{eq:spad_mean}. For the PD, the ML estimator uses its samples $\rho_{\mathrm{pd}}\zeta x(nT_s-\tau)$ in Gaussian noise, and its RMSE is averaged over $300$ frames. The analytical curves in Figs.~\ref{fig:count}(b)--\ref{fig:tradeoff} are averaged over $10^4$ blocks of $16$-QAM data.

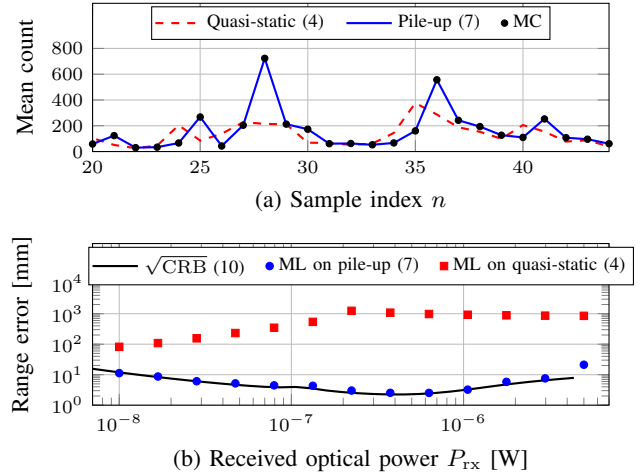
\begin{figure}[t]
\centering
\begin{tikzpicture}
\begin{axis}[
    name=prof,
    height=0.40\linewidth, width=0.95\linewidth,
    xlabel={(a) Sample index $n$},
    ylabel={Mean count},
    xmin=20, xmax=44, ymin=0, ymax=1150,
    ytick={0,200,400,600,800},
    xlabel style={font=\small}, ylabel style={font=\small},
    xticklabel style={font=\scriptsize}, yticklabel style={font=\scriptsize},
    legend style={font=\scriptsize, at={(0.5,0.98)}, anchor=north, inner sep=1pt},
    legend columns=3, legend style={/tikz/every even column/.append style={column sep=6pt}},
    grid=major,
]
\addplot[red, thick, dashed, no markers] table[x=n, y=c_steady] {new_fig_data/count_profile.dat};
\addlegendentry{Quasi-static \eqref{eq:spad_mean}}
\addplot[blue, thick, no markers] table[x=n, y=c_complete] {new_fig_data/count_profile.dat};
\addlegendentry{Pile-up \eqref{eq:qmem}}
\addplot[black, only marks, mark=*, mark size=1.2pt] table[x=n, y=c_mc] {new_fig_data/count_profile.dat};
\addlegendentry{MC}
\end{axis}
\begin{loglogaxis}[
    at={(prof.below south west)}, anchor=above north west, yshift=-2pt,
    height=0.42\linewidth, width=0.95\linewidth,
    xlabel={(b) Received optical power $P_{\mathrm{rx}}$ [W]},
    ylabel={Range error [mm]},
    xmin=7e-9, xmax=7e-6, ymin=1, ymax=2e5,
    ytick={1,10,100,1000,10000},
    xlabel style={font=\small}, ylabel style={font=\small},
    xticklabel style={font=\scriptsize}, yticklabel style={font=\scriptsize},
    legend style={font=\scriptsize, at={(0.52,0.98)}, anchor=north, inner sep=1pt, row sep=-1pt},
    legend columns=3, legend cell align=left, legend style={/tikz/every even column/.append style={column sep=6pt}},
    grid=major, unbounded coords=jump,
]
\addplot[black, thick, no markers, restrict expr to domain={\thisrow{eps}}{0.1:41}] table[x=Prx, y=crb_mm] {new_fig_data/fig3_agg_Na1024.dat};
\addlegendentry{$\sqrt{\mathrm{CRB}}$ \eqref{eq:crb}}
\addplot[blue, only marks, mark=*, mark size=1.4pt] table[x=Prx, y=rmse_pu_mm] {new_fig_data/rangewalk_rmse_mc.dat};
\addlegendentry{ML on pile-up \eqref{eq:qmem}}
\addplot[red, only marks, mark=square*, mark size=1.4pt] table[x=Prx, y=rmse_qs_mm] {new_fig_data/rangewalk_rmse_mc.dat};
\addlegendentry{ML on quasi-static \eqref{eq:spad_mean}}
\end{loglogaxis}
\end{tikzpicture}
\caption{(a) Mean array count per sample and (b) range RMSE versus received optical power for $N_a=1024$.}
\label{fig:count}
\end{figure}

\begin{figure}[t]
\centering
\begin{tikzpicture}
\begin{loglogaxis}[
    height=0.50\linewidth, width=0.95\linewidth,
    xlabel={Received optical power $P_{\mathrm{rx}}$ [W]},
    ylabel={$\sqrt{\mathrm{CRB}(D)}$ [mm]},
    xmin=3e-9, xmax=1e-5, ymin=0.5, ymax=100,
    xlabel style={font=\small}, ylabel style={font=\small},
    xticklabel style={font=\scriptsize}, yticklabel style={font=\scriptsize},
    legend style={font=\scriptsize, at={(0.5,1.03)}, anchor=north, inner sep=1pt, row sep=-1pt},
    legend columns=3, legend cell align=left, legend style={/tikz/every even column/.append style={column sep=4pt}},
    grid=major,
]
\addplot[red, thick, no markers, restrict expr to domain={\thisrow{eps}}{0:41}] table[x=Prx, y=crb_mm] {new_fig_data/fig3_agg_Na256.dat};
\addlegendentry{$N_a=256$}
\addplot[blue, thick, no markers, restrict expr to domain={\thisrow{eps}}{0:41}] table[x=Prx, y=crb_mm] {new_fig_data/fig3_agg_Na1024.dat};
\addlegendentry{$N_a=1024$}
\addplot[black!60!green, thick, no markers] table[x=Prx, y=crb_mm] {new_fig_data/fig3_agg_Na4096.dat};
\addlegendentry{$N_a=4096$}
\addplot[purple, thick, no markers] table[x=Prx, y=crb_mm] {new_fig_data/fig3_agg_Na8192.dat};
\addlegendentry{$N_a=8192$}
\addplot[blue, thick, dashed, no markers, restrict expr to domain={\thisrow{eps}}{0:41}] table[x=Prx, y=crb_steady_mm] {new_fig_data/fig3_agg_Na1024.dat};
\addlegendentry{Quasi-static bound}
\addplot[red, only marks, mark=*, mark size=1.3pt, forget plot] table[x=Prx, y expr=\thisrow{rmse_R}*1000] {new_fig_data/fig3_markers_Na256.dat};
\addplot[blue, only marks, mark=*, mark size=1.3pt, forget plot] table[x=Prx, y expr=\thisrow{rmse_R}*1000] {new_fig_data/fig3_markers_fine.dat};
\addplot[black!60!green, only marks, mark=*, mark size=1.3pt, forget plot] table[x=Prx, y expr=\thisrow{rmse_R}*1000] {new_fig_data/fig3_markers_Na4096.dat};
\addplot[purple, only marks, mark=*, mark size=1.3pt, forget plot] table[x=Prx, y expr=\thisrow{rmse_R}*1000] {new_fig_data/fig3_markers_Na8192.dat};
\addlegendimage{black, only marks, mark=*}
\addlegendentry{ML on \eqref{eq:qmem}}
\end{loglogaxis}
\end{tikzpicture}
\caption{Range CRB versus received optical power for different array sizes.}
\label{fig:crb}
\end{figure}
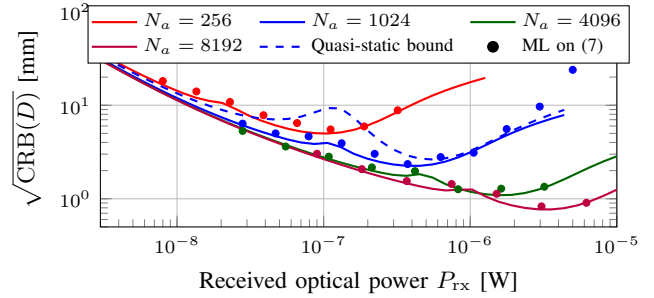

Fig.~\ref{fig:count} compares the two count models in \eqref{eq:spad_mean} and \eqref{eq:qmem} at $T_s=\tau_d$ and shows their impact on ranging. Fig.~\ref{fig:count}(a) shows the mean count of samples $n=20$ to $44$ of one data block at $P_{\mathrm{rx}}=0.41~\mu$W, i.e., in saturation ($\epsilon\approx3.4$). For this block, the pile-up model gives $N_a\tilde q[n]$ from \eqref{eq:qmem}, the quasi-static model gives $\mu_a(\lambda[n])$ from \eqref{eq:spad_mean}, and the MC markers average $y[n]$ over $4000$ photon realizations. The pile-up model follows the MC, whereas the quasi-static model misses the sample-to-sample variations. The reason is that, at $T_s=\tau_d$, the count in an interval depends on the photons of the preceding interval, whereas the quasi-static model assumes the same rate in both. After a dark interval, few pixels are blind, so that a bright interval yields more counts than \eqref{eq:spad_mean} predicts. For instance, after the intensity dip at $n=27$, $723$ pixels fire at $n=28$, whereas \eqref{eq:spad_mean} never exceeds its maximum $N_a/e\approx377$.

Fig.~\ref{fig:count}(b) shows how this mismatch affects ranging, from the unsaturated regime ($\epsilon\approx0.1$) to deep saturation ($\epsilon\approx40$). From the same counts, $\bm\theta$ is estimated by the ML estimator \eqref{eq:ml} based on \eqref{eq:qmem} and by the one based on \eqref{eq:spad_mean}, and the markers give the resulting range RMSE. The estimator based on \eqref{eq:qmem} follows $\sqrt{\mathrm{CRB}}$ from \eqref{eq:crb} except at the highest powers. The RMSE of the estimator based on \eqref{eq:spad_mean}, in contrast, is $0.08$ to $1.2$~m, up to $400$ times larger, and does not follow the trend of the bound, i.e., it grows with the received power in the unsaturated regime, where the bound decreases. The bias of this estimator, i.e., the mean of its range error, is more than $98\%$ of this RMSE, which reveals the power-dependent range bias predicted in Section~\ref{sec:system}.

Fig.~\ref{fig:crb} shows $\sqrt{\mathrm{CRB}}$ from \eqref{eq:crb} for four array sizes, the quasi-static bound, i.e., \eqref{eq:crb} evaluated with the quasi-static model \eqref{eq:spad_mean}, $\tilde q[n]=\mu_a(\lambda[n])/N_a$, for $N_a=1024$, and the RMSE of the ML estimator \eqref{eq:ml} based on \eqref{eq:qmem}. The estimator is unbiased, i.e., its mean range error is zero within about two standard errors of the MC average, and up to $\epsilon\approx10$ its RMSE exceeds $\sqrt{\mathrm{CRB}}$ by $9\%$ on average. The minimum of the bound improves roughly as $N_a^{-1/2}$, from $5.0$ to $0.8$~mm for $N_a=256$ to $8192$, and lies at $P_{\mathrm{rx}}\approx0.11$, $0.45$, $1.6$, and $3.6~\mu$W, respectively, i.e., near $\epsilon\approx3.5$ for all array sizes, beyond which the bound increases. Deep in saturation, the bound \eqref{eq:crb} becomes loose. The quasi-static bound is up to $2.6$ times larger at its blind spot $\epsilon\approx1$, where $\mu_a'=0$, and within about $25\%$ of \eqref{eq:crb} from the optimum onward. This bound holds only if the counts follow \eqref{eq:spad_mean}; for the actual counts, the estimator based on \eqref{eq:spad_mean} is biased, as discussed for Fig.~\ref{fig:count}(b).

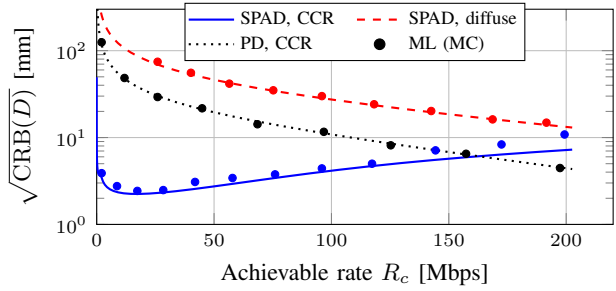
\begin{figure}[t]
\centering
\begin{tikzpicture}
\begin{semilogyaxis}[
    height=0.5\linewidth, width=0.95\linewidth,
    xlabel={Achievable rate $R_c$ [Mbps]},
    ylabel={$\sqrt{\mathrm{CRB}(D)}$ [mm]},
    xmin=0, xmax=220, ymin=1, ymax=300,
    xlabel style={font=\small}, ylabel style={font=\small},
    xticklabel style={font=\scriptsize}, yticklabel style={font=\scriptsize},
    legend style={font=\scriptsize, at={(0.5,1.03)}, anchor=north, inner sep=1pt, row sep=-1pt},
    legend columns=2, legend cell align=left, legend style={/tikz/every even column/.append style={column sep=6pt}},
    grid=major,
]
\addplot[blue, thick, no markers] table[x=Rc_Mbps, y=crb_ccr_mm] {new_fig_data/tradeoff_complete_B30.dat};
\addlegendentry{SPAD, CCR}
\addplot[red, thick, dashed, no markers] table[x=Rc_Mbps, y=crb_diff_mm] {new_fig_data/tradeoff_complete_B30.dat};
\addlegendentry{SPAD, diffuse}
\addplot[black, thick, dotted, no markers] table[x=Rc_Mbps, y=crb_pd_mm] {new_fig_data/tradeoff_complete_B30.dat};
\addlegendentry{PD, CCR}
\addplot[blue, only marks, mark=*, mark size=1.4pt, forget plot] table[x=Rc_Mbps, y=rmse_mm] {new_fig_data/tradeoff_mc_ccr.dat};
\addplot[red, only marks, mark=*, mark size=1.4pt, forget plot] table[x=Rc_Mbps, y=rmse_mm] {new_fig_data/tradeoff_mc_diff.dat};
\addplot[black, only marks, mark=*, mark size=1.4pt, forget plot] table[x=Rc_Mbps, y=rmse_mm] {new_fig_data/tradeoff_mc_pd_L30.dat};
\addlegendimage{black, only marks, mark=*}
\addlegendentry{ML (MC)}
\end{semilogyaxis}
\end{tikzpicture}
\caption{Range CRB versus achievable rate for $N_a=1024$.}
\label{fig:tradeoff}
\end{figure}

Fig.~\ref{fig:tradeoff} shows the resulting root CRB-rate region~\cite{xiong2023fundamental}, traced by sweeping $P_t$ up to $P_{\max}$, together with the RMSE of the corresponding ML estimators. The rate grows monotonically with the power, whereas the CCR range CRB attains its minimum at $P_t^{\mathrm s}\approx0.22$~W, i.e., at $R_c\approx17$~Mbps and $\sqrt{\mathrm{CRB}}\approx2.2$~mm. Beyond this point, a higher power increases the rate at the cost of ranging accuracy, and at $P_{\max}$, where $R_c\approx203$~Mbps, the bound is about $3$ times larger than at $P_t^{\mathrm s}$. The CRB thus does not decrease beyond $P_t^{\mathrm s}$, as \eqref{eq:pstar} requires. The diffuse target remains unsaturated, so that its CRB decreases monotonically and $P_t^\star=P_{\max}$. The PD bound is monotone and is attained within $10\%$. The SPAD is about $17$ times more precise at $P_t^{\mathrm s}$, and the two bounds cross only near $160$~Mbps, beyond which the PD is more precise. Hence, the trade-off is a receiver-physics effect.

\section{Conclusion}\label{sec:concl}

We analyzed ranging with a SPAD array in DCO-OFDM O-ISAC. Sampling the array once per dead time, we modeled the array count in each sampling interval as a binomial random variable that accounts for the intensity variation within a dead time, derived the range CRB, and validated it by MC simulation of the photon counting. At this sampling rate, the SPAD count model of OWC, which assumes a constant photon rate over a dead time, causes a power-dependent range bias. Moreover, saturation makes the ranging accuracy of the SPAD array best at a sensing-optimal transmit power, whereas a linear receiver benefits from more power. Under a rate constraint, the optimal transmit power is then the larger of this power and the minimum power required by the target rate.

\bibliographystyle{IEEEtran}
\bibliography{refs}

@article{wen2023fsoisac,
  author  = {Y. Wen and F. Yang and J. Song and Z. Han},
  title   = {{Free-space optical integrated sensing and communication based on DCO-OFDM: Performance metrics and resource allocation}},
  journal = {IEEE Internet Things J.},
  volume  = {12}, number = {2}, pages = {2158--2173}, month = jan, year = {2025}
}

@inproceedings{nikolaidou2025oisac,
  author = {D. Nikolaidou and D. Bozanis and P. Sarigiannidis and G. K. Karagiannidis and R. Schober and C. Masouros and V. K. Papanikolaou},
  title  = {{On the performance analysis of optical wireless integrated sensing and communications}},
  booktitle = {Proc. IEEE Int. Conf. Commun. Workshops (ICC Wkshps)},
  address = {Glasgow, U.K.},
  pages = {1--6}, month = may, year = {2026}
}

@article{cui2024retroisac,
  author  = {Y. Cui and C. Chen and Y. Cai and Z. Zeng and M. Liu and J. Ye and S. Shao and H. Haas},
  title   = {{Retroreflective optical ISAC using OFDM: Channel modeling and performance analysis}},
  journal = {Opt. Lett.},
  volume  = {49}, number = {15}, pages = {4214--4217}, month = aug, year = {2024}
}

@article{zhang2025led,
  author  = {R. Zhang and Y. Shao and M. Li and L. Lu and Y. C. Eldar},
  title   = {{Optical integrated sensing and communication with light-emitting diode}},
  journal = {IEEE Internet Things J.},
  volume  = {12}, number = {9}, pages = {12896--12911}, month = may, year = {2025}
}

@article{huang2023jocn,
  author  = {S. Huang and Y. Li and C. Chen and M. {Dehghani Soltani} and R. K. Henderson and M. Safari and H. Haas},
  title   = {{Performance analysis of SPAD-based optical wireless communication with OFDM}},
  journal = {J. Opt. Commun. Netw.},
  volume  = {15}, number = {3}, pages = {174--186}, month = mar, year = {2023}
}

@article{sarbazi2017jlt,
  author  = {E. Sarbazi and M. Safari and H. Haas},
  title   = {{Statistical modeling of single-photon avalanche diode receivers for optical wireless communications}},
  journal = {IEEE Trans. Commun.},
  volume  = {66}, number = {9}, pages = {4043--4058}, month = sep, year = {2018}
}

@article{komine2004vlc,
  author  = {T. Komine and M. Nakagawa},
  title   = {{Fundamental analysis for visible-light communication system using LED lights}},
  journal = {IEEE Trans. Consum. Electron.},
  volume  = {50}, number = {1}, pages = {100--107}, month = feb, year = {2004}
}

@book{kay1993,
  author    = {S. M. Kay},
  title     = {{Fundamentals of Statistical Signal Processing: Estimation Theory}},
  publisher = {Prentice-Hall}, address = {Englewood Cliffs, NJ, USA}, year = {1993}
}

@article{liu2022crb,
  author  = {F. Liu and Y.-F. Liu and A. Li and C. Masouros and Y. C. Eldar},
  title   = {{Cram\'er--Rao bound optimization for joint radar-communication beamforming}},
  journal = {IEEE Trans. Signal Process.},
  volume  = {70}, pages = {240--253}, year = {2022}
}

@article{liu2018dfrc,
  author  = {F. Liu and L. Zhou and C. Masouros and A. Li and W. Luo and A. Petropulu},
  title   = {{Toward dual-functional radar-communication systems: Optimal waveform design}},
  journal = {IEEE Trans. Signal Process.},
  volume  = {66}, number = {16}, pages = {4264--4279}, month = aug, year = {2018}
}

@article{xiong2023fundamental,
  author  = {Y. Xiong and F. Liu and Y. Cui and W. Yuan and T. X. Han and G. Caire},
  title   = {{On the fundamental tradeoff of integrated sensing and communications under Gaussian channels}},
  journal = {IEEE Trans. Inf. Theory},
  volume  = {69}, number = {9}, pages = {5723--5751}, month = sep, year = {2023}
}

@article{xu2025optisac,
  author  = {C. Xu and C. Masouros and S. Sugiura and P. Petropoulos and R. G. Maunder and L.-L. Yang and H. Haas and L. Hanzo},
  title   = {{Integrated positioning and communication relying on wireless optical OFDM}},
  journal = {IEEE J. Sel. Areas Commun.},
  volume  = {43}, number = {5}, pages = {1721--1737}, month = may, year = {2025}
}

@article{coates1968,
  author  = {P. B. Coates},
  title   = {{The correction for photon `pile-up' in the measurement of radiative lifetimes}},
  journal = {J. Phys. E: Sci. Instrum.},
  volume  = {1}, number = {8}, pages = {878--879}, month = aug, year = {1968}
}

@article{rapp2021highflux,
  author  = {J. Rapp and Y. Ma and R. M. A. Dawson and V. K. Goyal},
  title   = {{High-flux single-photon lidar}},
  journal = {Optica},
  volume  = {8}, number = {1}, pages = {30--39}, month = jan, year = {2021}
}

@article{goldsmith1997mqam,
  author  = {A. J. Goldsmith and S.-G. Chua},
  title   = {{Variable-rate variable-power MQAM for fading channels}},
  journal = {IEEE Trans. Commun.},
  volume  = {45}, number = {10}, pages = {1218--1230}, month = oct, year = {1997}
}

@article{rapp2019deadtime,
  author  = {J. Rapp and Y. Ma and R. M. A. Dawson and V. K. Goyal},
  title   = {Dead time compensation for high-flux ranging},
  journal = {IEEE Trans. Signal Process.},
  volume  = {67},
  number  = {13},
  pages   = {3471--3486},
  month   = jul,
  year    = {2019}
}

@inproceedings{gupta2019flooded,
  author    = {A. Gupta and A. Ingle and A. Velten and M. Gupta},
  title     = {Photon-flooded single-photon {3D} cameras},
  booktitle = {Proc. IEEE/CVF Conf. Comput. Vis. Pattern Recognit. (CVPR)},
  address   = {Long Beach, CA, USA},
  pages     = {6763--6772},
  month     = jun,
  year      = {2019}
}

@article{wu2025spadcrlb,
  author  = {H. Wu and Y. Wang and S. Sun and L. Zhao and L. Tong and L. Shen and J. Zhu},
  title   = {{Performance bounds of ranging precision in SPAD-based dToF LiDAR}},
  journal = {Sensors},
  volume  = {25}, number = {19}, month = oct, year = {2025}
}

\end{document}